\RequirePackage{fix-cm}
\documentclass[twocolumn]{svjour3}       
\smartqed  
\usepackage[utf8]{inputenc}

\usepackage{graphicx}
\usepackage{amsmath}
\usepackage{amsfonts}
\usepackage{subfigure}
\usepackage{algorithm}
\usepackage{algorithmic}
\usepackage[version=3]{mhchem}
\usepackage{rotating}
\usepackage{dsfont}
\usepackage{caption}
\usepackage{float}
\usepackage[numbers, sort&compress]{natbib}
\usepackage{amssymb}
\DeclareMathOperator{\E}{\mathbb{E}}
\DeclareMathOperator*{\argmax}{\arg\!\max}

\begin{document}

\title{An automatic adaptive method to combine summary statistics in approximate Bayesian computation 
}

\titlerunning{Combining summary statistics in approximate Bayesian computation}        

\author{Jonathan U. Harrison         \and
        Ruth E. Baker 
}

\authorrunning{Jonathan Harrison 
\and 
Ruth Baker
} 

\institute{Jonathan U. Harrison \at
              Mathematical Institute, \\ University of Oxford \\
              \email{harrison@maths.ox.ac.uk}           
           \and
           Ruth E. Baker \at
              Mathematical Institute, \\ University of Oxford
}

\date{Received: date / Accepted: date}

\maketitle

\begin{abstract}
To infer the parameters of mechanistic models with intractable likelihoods, techniques such as approximate Bayesian computation (ABC) are increasingly being adopted.
One of the main disadvantages of ABC in practical situations, however, is that parameter inference must generally rely on summary statistics of the data.
This is particularly the case for problems involving high-dimensional data, such as biological imaging experiments.
However, some summary statistics contain more information about parameters of interest than others,
and it is not always clear how to weight their contributions within the ABC framework.
We address this problem by developing an automatic, adaptive algorithm that chooses weights for each summary statistic.
Our algorithm aims to maximize the distance between the prior and the approximate posterior by automatically adapting the weights within the ABC distance function.
Computationally, we use a nearest neighbour estimator of the distance between distributions. 
We justify the algorithm theoretically based on properties of the nearest neighbour distance estimator.
To demonstrate the effectiveness of our algorithm, we apply it to a variety of test problems, including several stochastic models of biochemical reaction networks, and a spatial model of diffusion, and compare our results with existing algorithms.

\keywords{Approximate Bayesian Computation \and Summary Statistics \and Sequential Monte Carlo \and Likelihood-free}
\end{abstract}

\section{Introduction}

When using quantitative models to explore biological or physical phenomena,
it is crucial to be able to estimate parameters of these models and account appropriately for uncertainty in both the parameters and model predictions.
Bayesian statistics offers a wealth of tools in this regard \cite{hines2015primer,wilkinson2009stochastic}.
Bayes' theorem gives us that the posterior, $p(\theta | D)$, of parameters, $\theta$, given data, $D$, is proportional to a prior, $\pi(\theta)$,
on the parameters multiplied by the likelihood, $p(D| \theta)$, of data, $D$, given those parameters: $p(\theta | D) \propto p(D|\theta) p(\theta)$.
The prior represents our beliefs about the parameters prior to observing the data, the likelihood gives the probability of observing the data, given a certain set of parameters, and these result in the posterior,
which returns updated beliefs about the parameters after having observed the data.

However, much of the current theory surrounding the generation of posterior distributions for parameter inference relies on being able to evaluate the likelihood of the data given the parameters of a model. 
In practice, for a large class of mechanistic models the likelihood is not tractable, either due to computational or analytical complexity.
Therefore, the use of likelihood-free methods for inference, including approximate Bayesian computation (ABC) \cite{pritchard1999population,beaumont2002approximate,beaumont2010approximate,turner2012tutorial,sunnaaker2013approximate}, 
indirect inference \cite{gourieroux1993indirect}, synthetic likelihoods \cite{wood2010statistical,price2017bayesian}, particle Markov Chain Monte Carlo (pMCMC) \cite{andrieu2009pseudo,andrieu2010particle,golightly2011bayesian,owen2015scalable}, 
expectation propogation \cite{barthelme2014expectation}, and other similar methods, has become widespread \cite{hartig2011statistical}.
In particular, ABC has been widely adopted due to its ease of understanding and implementation.

\subsection{Approximate Bayesian computation}
Suppose we wish to infer a posterior distribution over parameters $\theta$ of a generative model such that we can simulate from $\textbf{x} \sim f(\textbf{x} | \theta) $.
In ABC, parameters $\theta$ are drawn from a prior, $\pi(\theta)$, and data, $\mathbf{x^*}$, is simulated from the generative model using those parameters, such that $\mathbf{x^*} \sim f(\mathbf{x}|\theta)$.
The distance between the simulated dataset, $\mathbf{x^*}$, and the real data, $\mathbf{y}$, is calculated using a distance function $d(\mathbf{x^*},\mathbf{y})$.
If this distance is less than a certain tolerance, $\epsilon$, then the parameters $\theta$ can be accepted into the approximate posterior sample.
Choice of the tolerance $\epsilon$ can be avoided, to some extent, by simulating a large number, $N$, of parameter samples and datasets,
calculating the corresponding distances for these and accepting the proportion $\alpha$ that lie closest to the real data.
We will use this approach in this work.

In cases where the prior and posterior distributions are very different, the rejection sampling version of ABC described above can have very low acceptance rates.
Algorithm 1 summarizes how samples from an approximate posterior can be generated via a more efficient version of ABC using sequential Monte Carlo techniques, known as ABC-SMC \cite{toni2009approximate,sisson2007sequential, del2006sequential}.
Importance sampling is used iteratively so that instead of sampling repeatedly from the prior, parameters are sampled from an approximate posterior at each generation of the algorithm.
A weight must be given to each sample to correct for the fact that it is not drawn from the prior.

\begin{algorithm}[h!] \label{abc-smc}
\floatname{algorithm}{Algorithm 1}
\renewcommand{\thealgorithm}{}
\caption{ABC-SMC}

\begin{algorithmic}[1]
\STATE Set generation index $t=1$. \\
\FOR{$i$ = 1 to $N$} 
  \STATE Sample from prior $\theta^{**} \sim \pi(\theta)$. \\
  \STATE Simulate dataset $\mathbf{x}^{*i} \sim f(\mathbf{x}|\theta^{**}) $ and calculate distance $d(s(\mathbf{y}),s(\mathbf{x}^{*i}))$. 
  \STATE Set $\theta^i_t = \theta^{**}.$ \\
  Calculate the weight $v^i_t$ for particle $\theta^i_t$ as $ v^i_t = 1 $
  \STATE Normalize the particle weights ${v_t^i}$. \\
\ENDFOR
 \STATE Select the proportion $\alpha$ of samples closest to the real data to keep and reject the rest, resulting in $M=\left \lfloor{\alpha N}\right \rfloor $ samples.
\FOR{$t$ = 2 to $T$}
\FOR{$i$ = 1 to $N$}
  \STATE \label{step3} Sample $\theta^{*}$ from previous population ${\theta_{t-1}^{i}}$ with weights $v_{t-1}$. \\
  Peturb $\theta^*$ to give $\theta^{**} \sim K_t(\theta|\theta^*)$.
  If $\pi(\theta^{**}) = 0$, return to step 11.
  \STATE \label{step4} Simulate dataset $\mathbf{x}^{*i} \sim f(\mathbf{x}|\theta^{**}) $ and calculate distance $d(s(\mathbf{y}),s(\mathbf{x}^{*i}))$. 
  \STATE Set $\theta^i_t = \theta^{**}.$ \\
  Calculate the weight $v^i_t$ for particle $\theta^i_t$ via
  \begin{equation*}
    v^i_t = \frac{\pi (\theta_t^i)}{\sum_{j=1}^M v_{t-1}^j K_t (\theta_{t-1}^j, \theta_t^i)}.
  \end{equation*}
  \STATE Normalize the particle weights ${v_t^i}$. \\
\ENDFOR
\STATE Select the proportion $\alpha$ of samples closest to the real data to keep and reject the rest, resulting in $M=\left \lfloor{\alpha N}\right \rfloor$ samples.
\ENDFOR
\STATE \textbf{return} $\{\theta_T^i\}_{i=1}^M$, $\{v_T^i\}_{i=1}^M$
\end{algorithmic}
\end{algorithm}

\subsection{Choice of summary statistics}
Suppose we are interested in inferring multi-dimensional parameters for a model that we can simulate, but cannot evaluate the likelihood directly. 
In many practical circumstances, the data (either collected experimentally or simulated from the \textit{in silico} model) will be very high dimensional. 
High-dimensional data poses difficulties within the ABC framework, as it is difficult to sensibly estimate when the output of a particular simulation is `close' to the data.
Even taking account of domain expertise, it can be hard to determine which features of the data are important.
This issue of comparing high-dimensional data is further compounded using stochastic models where there is noise in the process model in addition to measurement noise.
Repeatedly drawing from a stochastic model with the same parameter values can give vastly different outputs.

As such, it is often necessary to work with a lower-dimensional vector of summary statistics, $s(\mathbf{x})$, of the data, such that we require the distance between summary statistics is less than the tolerance, $d(s(\mathbf{x^*}),s(\mathbf{y})) < \epsilon$.
Examples of these summary statistics may be data points within a time series, an average transition time between different states of a system, or the moments of a certain species within a model.
However, not all summary statistics are equally informative about the posterior.
Common practice is to combine summary statistics based on some heuristic approach, such as the weighting the contribution of each summary statistic according to its standard deviation. 
However, it is not clear whether these heuristic approaches result in optimal weighting of the various summary statistics available.
As such, the aim of this work is to provide an automated and adaptive method for determining the weighting of available summary statistics in order to optimize the quality of the resulting posterior.

Previous work has also considered how to weight or select summary statistics for ABC. 
\citet{fearnhead2012constructing} developed a popular method to find informative linear combinations of summary statistics by fitting a regression for each model parameter.
Another successful approach is the subset selection method of \citet{barnes2012considerate}, which uses an approximate sufficiency criterion to select a subset of summary statistics on which to base inference.
We explore these methods in further detail in Section \ref{DimReduction}.

A genetic algorithm has been used to choose the weights of different summary statistics \cite{jung2011choice}. 
This genetic algorithm attempts to optimize the mean squared error (MSE) of the posterior samples from the true parameter, however this is generally not known in practice.
A method for adaptively choosing summary statistic weights for ABC based on the scale of the summary statistics has also been investigated \cite{prangle2015adapting}.
The median absolute deviation, a measure of spread of a statistic, is used for the scaling. 
The aim is that all summary statistics contribute equally to the distance function. 
In practice, however, this may not be the most desirable choice, as some summary statistics are clearly more informative than others.

Recently, \citet{singh2018multi} have proposed a multi-armed bandit problem approach to selecting summary statistics for ABC.
Approaches using machine learning tools such as random forests \cite{pudlo2015reliable} to aid model selection and neural networks to form a parameteric model of the posterior \cite{papamakarios2016fast, papamakarios2018sequential} have also been investigated.
Other work has avoided using summary statistics at all by considering Wasserstein distance between full data sets \cite{bernton2017inference}.

In this work, we approach the problem from the point of view of finding the right distance function, adapted to information contained in the summary statistics, rather than selecting a certain subset of summary statistics.
We provide an automatic algorithm that adaptively selects weights for each summary statistic within the ABC distance function.

\subsection{Outline} \label{Outline}
Our contribution in this work is to present a flexible, novel framework for improving inference with ABC by adapting the weights of different summary statistics to maxmize the gain in posterior information from a dataset.
This is helpful for avoiding bias and variance from redundant information in data (such as would be the case when including a summary statistic that is uncorrelated with the parameters of interest).
A further advantage of our work is that it can alleviate the burden of designing and selecting summary statistics `by hand', since a large collection of summaries can be used and weighted appropriately via our procedure.
It is also possible to combine our framework with existing dimensionality reduction techniques for summary statistics in ABC (see Section \ref{DimReduction}).

We outline in Section \ref{weights_algorithm} our adaptive algorithm for combining summary statistics in an ABC framework.
We provide theoretical justification for the algorithm in Section \ref{justification} and demonstrate that, in the appropriate limit, we obtain convergence to the posterior distribution. 
To demonstrate the utility of our algorithm, we apply it to several test problems based on biochemical reaction networks in Section \ref{Examples}.
We compare results of parameter inference using our algorithm against benchmark results
from applying ABC-SMC using other choices of weights for the summary statistics.
Finally, in Section \ref{Summary}, we summarize the work presented in this article and 
compare our methodology for combining summary statistics with other techniques
in the literature that are based on dimensionality reduction of a set of summary statistics.

\section{An algorithm for automatic weighting of summary statistics} \label{weights_algorithm}

In order to use ABC-SMC (see Algorithm 1), we must specify a function to measure the distance between simulated and real datasets.
Suppose we take a weighted Euclidean distance as our ABC distance function such that $$d_{\mathbf{w}}(s(\mathbf{x_1}),s(\mathbf{x_2})) = \sum_{i=1}^{\kappa} w_i (s_{1 i} - s_{2 i})^2, $$
where $s(\mathbf{x}) = (s_1, \ldots, s_{\kappa}) \in \mathbb{R}^{\kappa}$ is a vector of summary statistics and the sum over $i$ is taken over all the summary statistics considered.
This distance function is a reasonable and flexible choice commonly used in the literature \cite{mckinley2009inference}.
It is these distance weights, $w_i$, that control how the summary statistics are combined in this case.
Given simulated pairs of parameter samples and datasets, we find weights, $\mathbf{w} = (w_1, \ldots, w_\kappa) \in E \subset \mathbb{R}^{\kappa}$, that maximize a distance between the prior and the posterior that
represents the maximum possible gain in information about the parameters from the given data.
Constructing the weights in this way allows us to account for the scale of the summary statistics, as well as their relative contribution to a posterior. 

\subsection{Adaption of weights} \label{Adaption_method}
We seek to optimize the weights, $\mathbf{w}$, so that we can place less emphasis on summary statistics that are not informative for the posterior,
but also scale summary statistics appropriately so that we do not neglect to obtain information about certain parameters.
We do this within the ABC-SMC framework \cite{del2006sequential,sisson2007sequential,toni2009approximate} given in Algorithm 1. 
We outline our proposed methodology in Algorithm 2.

At each generation, we search for the weights, $\mathbf{w}$, of the distance function that maximize the distance between the prior and resulting posterior,
given $N$ ABC samples from the model for different $\theta$ values. 
This distance between prior and posterior gives a measure of the information gain in moving from the prior to the posterior.

\subsection{Distance between distributions}
We use the Hellinger distance to measure the discrepency between the prior and posterior, and so to detect the optimality of our posterior.
The Hellinger distance is defined, for distributions $P$ and $Q$, with densities $p$ and $q$, respectively, as
\begin{align*}
 H^{2}(P,Q) &= {\frac{1}{2}}\int \left({\sqrt{p(x)}}-{\sqrt{q(x)}}\right)^{2}\,\mathrm{d}x \\
	    &= 1 - \int \sqrt{p(x)q(x)} \, \mathrm{d} x.
\end{align*}
Alternative measures of distance between distributions such as the Euclidean distance or Kullback-Leibler (or KL) divergence can be used.
In our experience, the Hellinger distance performs better than alternatives, particularly for robustly identifying relatively small differences between posterior distributions when weights are optimized,
an observation that is supported by other work \cite{jones2015inference}.
In addition, the Hellinger distance is finite when comparing distributions with different support (unlike the KL divergence). 
This property is desirable when comparing a broad prior with a posterior distribution where we have gained some knowledge of parameter space and can exclude certain regions.

\subsubsection{Nearest neighbour distance estimator} \label{knn_estimator}
To estimate the distance between two distributions, based on samples from these distributions, we use a $k$ nearest neighbour estimator \cite{poczos2011estimation,poczos2012nonparametric,sutherland2012kernels}
developed to describe a family of distances between distributions known as $\alpha$ divergences, of which the Hellinger distance is a special case. 
Suppose we have two probability distributions, $P$ and $Q$ with densities $p(x)$, and $q(x)$, and are interested in the distance between these.
We suppose that we have some samples, $X_{1:n}$ and $Y_{1:n}$ from $p$ and $q$.
If we define $ D_{\alpha} (p||q) = \int p^{\alpha}(x) q^{1-\alpha}(x) \, \mathrm{d} x$ for $\alpha \in \mathbb{R}$, then the Hellinger distance is
$$ D_h (p||q) = 1 - D_{1/2}(p||q). $$

The $k$ nearest neighbour estimator that we use depends only on distances between observations in a sample.
Let $\rho_k(i)$ be the Euclidean distance from the sample $X_i$ to its $k$th nearest neighbour in $X_{1:n}$.
Similarly, let $\nu_k(i)$ be the distance from $X_i$ to its $k$th nearest neighbour in the samples $Y_{1:n}$.
Then the estimator \cite{poczos2011estimation} is given by
\begin{equation} \label{d_hat}
 \hat{D_{\alpha}}\left(X_{1:n} || Y_{1:n}\right) = \frac{1}{n} \sum_{i=1}^n \left( \frac{(n-1)\rho_k(i)}{n\nu_k(i)} \right)^{1-\alpha} B_{k,\alpha},
\end{equation}
where $$B_{k,\alpha} = \frac{\Gamma(k)^2}{\Gamma(k-\alpha+1)\Gamma(k+\alpha-1)}.$$

At each generation of ABC-SMC we seek to find weights $\mathbf{w} \in E \subset \mathbb{R}^{\kappa}$ such that
\begin{equation}
\mathbf{w}^* = \argmax_{\mathbf{w} \in E} \left(1 - \hat{D}_{\alpha} \left( \{\xi_i\}_{i=1}^M || \{\theta_i\}_{i=1}^M \right) \right),
\end{equation}
where $\{\xi_i\}_{i=1}^M$ are samples from the prior distribution, and  $\{\theta_i\}_{i=1}^M$ are samples from the approximate posterior distribution, which depends on the summary statistic weights, $\mathbf{w}$.

To perform the optimization in weight space in our implementation of Algorithm 2, we use a constrained nonlinear optimizer, implemented via $\tt{fmincon}$ in MATLAB \cite{MATLAB:2016}.

\begin{algorithm}[h!]
\floatname{algorithm}{Algorithm 2}
\renewcommand{\thealgorithm}{}
\caption{Adaption of distance weights for ABC}
\label{protocol1}
\begin{algorithmic}[1]
\STATE Set generation index $t=1$. \\
\FOR{$i$ = 1 to $N$} 
  \STATE Sample from prior $\theta^{**} \sim \pi(\theta)$. \\
  \STATE Simulate dataset $\mathbf{x}^{*i} \sim f(\mathbf{x}|\theta^{**}) $. 
  \STATE Set $\theta^i_t = \theta^{**}.$ \\
  Calculate the weight $v^i_t$ for particle $\theta^i_t$ as $ v^i_t = 1 $.
\ENDFOR
\STATE Let $$L(\mathbf{w}) = 1 - \hat{D}_{\alpha} \left( \left\{ \xi^i \right\}_{i=1}^M || \left\{ \theta_t^i \right\}_{i=1}^M \right),$$
where $\left\{ \theta_t^i \right\}_{i=1}^M$ are the closest $M=\left \lfloor{\alpha N} \right \rfloor$ samples when ranked according to ABC distance from the pseudo dataset, $d_{\mathbf{w}}(s(\mathbf{y}),s(\mathbf{x}^{i*}))$.
\STATE Maxmize $L(\mathbf{w})$ as a function of summary statistic weights, $\mathbf{w}$.  
\STATE Keep the samples $\left\{ \theta_t^i \right\}_{i=1}^M$ corresponding to the maxmimum of $L(\mathbf{w})$ (i.e. the maximum distance between prior and approximate posterior).\\
\STATE Normalize the particle weights $\left\{ v_t^i \right\}_{i=1}^M$. \\
\FOR{$t$ = 2 to $T$}
\FOR{$i$ = 1 to $N$}
  \STATE \label{step3} Sample $\theta^{*}$ from previous population ${\theta_{t-1}^{i}}$ with weights $v_{t-1}$. \\
  Peturb $\theta^*$ to give $\theta^{**} \sim K_t(\theta|\theta^*)$.
  If $\pi(\theta^{**}) = 0$, return to step 11.
  \STATE \label{step4} Simulate dataset $\mathbf{x}^{*i} \sim f(\mathbf{x}|\theta^{**}) $. 
  \STATE Set $\theta^i_t = \theta^{**}.$ \\
  Calculate the weight $v^i_t$ for particle $\theta^i_t$ as
  \begin{equation*}
    v^i_t = \frac{\pi (\theta_t^i)}{\sum_{j=1}^M v_{t-1}^j K_t (\theta_{t-1}^j, \theta_t^i)}.
  \end{equation*}
\ENDFOR
\STATE Let $$L(\mathbf{w}) = 1 - \hat{D}_{\alpha}\left(\left\{ \xi^i \right\}_{i=1}^M || \left\{ \theta_t^i \right\}_{i=1}^M\right),$$
where $\left\{ \theta_t^i \right\}_{i=1}^M$ are the closest $M=\left \lfloor{\alpha N} \right \rfloor$ samples when ranked according to ABC distance from the pseudo dataset, $d_{\mathbf{w}}(s(\mathbf{y}),s(\mathbf{x}^{i*}))$.
\STATE Maxmize $L(\mathbf{w})$ as a function of summary statistic weights, $\mathbf{w}$.  
\STATE Keep the samples $\left\{ \theta_t^i \right\}_{i=1}^M$ corresponding to the maxmimum of $L(\mathbf{w})$ (i.e. the maximum distance between prior and approximate posterior).
\STATE Normalize the particle weights $\left\{ v_t^i \right\}_{i=1}^M$. \\
\ENDFOR
\STATE \textbf{return} $\{\theta_T^i\}_{i=1}^M$, $\{v_T^i\}_{i=1}^M$
\end{algorithmic}
\end{algorithm}

\section{Theoretical justification} \label{justification}

The estimator we use is a $k$ nearest neighbour estimator relying only on distances between observations in a sample, as described above in Section \ref{knn_estimator}.
We note that although the Hellinger distance, $D_h(p||q)$, is symmetric in $p$ and $q$, the estimator above in eq. \eqref{d_hat} is not.
By using the estimator from eq. \eqref{d_hat} and choosing $q$ as the distribution that depends on the parameters, $\mathbf{w}$, we are able to make strong assumptions about $p$ independent of the parameters, $\mathbf{w}$, and make weaker assumptions about $q$.
In the context of our algorithm for ABC, this allows us to treat $p$ as the prior and $q$ as the approximate posterior distribution.

We will require the following results:

\begin{lemma}
\citet{poczos2011estimation} \\
Suppose that $k\ge 2$ and that $\mathcal{M} = \text{supp}(p)$.
Assume that (a) $q$ is bounded above, (b) $p$ is bounded away from zero, (c) $p$ is uniformly Lebesgue approximable \footnote{\begin{definition}
(Uniformly Lebesgue-approximable function). Let $g \in L_1(E)$ for $E \subset \mathbb{R}^d$. $g$ is \textit{uniformly Lebesgue approximable} on $E$ if, for any sequence $R_n \rightarrow 0$ and any $\delta >0$, $\exists \, n=n_0(\delta) \in \mathbb{Z}^+ $ (independent of $x$) such that
if $n>n_0$, then for almost all $x \in E$,
$$g(x) - \delta < \frac{\int_{\mathcal{B}(x,R_n)\cap E} g(t) \mathrm{d}t}{\mathcal{V}(\mathcal{B}(x,R_n)\cap E)} < g(x) + \delta, $$
where $\mathcal{B}(x,R)$ is the closed ball around point $x \in \mathbb{R}^d$ with radius $R$, and $\mathcal{V}(\mathcal{B}(x,R))$ is the volume of the ball.
\end{definition}},
(d) $\exists \, \delta_0$ such that $\forall \delta \in (0,\delta_0) \, \int_{\mathcal{M}} H(x,p,\delta,1/2)p(x) \mathrm{d}x < \infty,$
(e) $\int_{\mathcal{M}}||x-y||^{\gamma} p(y) \mathrm{d}y < \infty$ for almost all $x \in \mathcal{M}$, $\int \int_{\mathcal{M}^2}||x-y||^{\gamma} p(y) p(x) \mathrm{d}y \mathrm{d} x < \infty$, where
$H(x,p,\delta,\psi) = \\ \sum_{j=0}^{k-1} \left(\frac{1}{j!}\right)^{\psi} \Gamma(1 - \alpha + j\psi) \left( \frac{p(x)+\delta}{p(x)-\delta}\right)^{j\psi} (p(x)-\delta)^{-(1-\alpha)}\left((1-\delta)\psi \right)^{-(1-\alpha)-j\psi}.$ \newline 
Then 
\begin{equation} \label{asymptoticly_unbiased}
\lim_{n \rightarrow \infty} \E\left[ \left(\hat{D}_h\left( X_{1:n} || Y_{1:n}\right) - D_h\left(p || q\right) \right)^2 \right] = 0.
\end{equation}
\end{lemma}
The proof of this results relies on constructing an integrable function as a bound such that Lebesgue's dominated convergence theorem can be applied.
See \citet{poczos2011estimation} for details.
Lemma 1 specifies $L_2$ consistency of the nearest neighbour estimator, which ensures that the estimates of the distance between $p$ and $q$ become more concentrated around the true values as more samples are used.

 \begin{lemma}
  Let $X_n$ and $X$ be random variables in $\mathbb{R}^d$. \newline
  If $$\lim_{n\rightarrow \infty} \E \left[ ||X_n - X ||^2\right] = 0,$$
  then, for any $\epsilon >0$, $$\lim_{n \rightarrow \infty} \mathbb{P} \left( ||X_n - X|| > \epsilon \right) = 0.$$ 
  That is, $L_2$ convergence implies convergence in probability.
 \end{lemma}
\begin{proof}
 By Chebyshev's inequality, for any $\epsilon>0$
 $$
 \mathbb{P} \left( ||X_n - X||>\epsilon \right) \le \frac{\E \left[ ||X_n -X||^2 \right]}{\epsilon^2}.
 $$
\end{proof}

\begin{theorem}
 Assume $E$ is finite, $|E| = \chi$.
 Assume the conditions of Lemma 1 hold for distributions $p$ and $q^{(\mathbf{w})}$.
 Assume that a unique $\mathbf{w}^*$ maximises $D_h\left(p || q^{(\mathbf{w})} \right)$ and arrange parameter values $\mathbf{w}^j$ for $j \in \{1, \dots , \chi\}$ in order such that they are descending in $D_h\left(p || q^{(\mathbf{w})} \right)$.
 That is $\mathbf{w}^1 = \mathbf{w}^*$, using $\mathbf{w}^2$ gives the next biggest value and so on.
 Then 
 \begin{equation}
\lim_{n \rightarrow \infty} \mathbb{P} \left( \argmax_{\mathbf{w} \in S} \hat{D}_h \left( X_{1:n} || Y_{1:n}^{(\mathbf{w})} \right) = \mathbf{w}^* \right) = 1.
 \end{equation}

\end{theorem}
\begin{proof}
 Let $\epsilon >0$.
 Take $\delta < D_h\left(p || q^{(\mathbf{w}^*)} \right) - D_h\left(p || q^{(\mathbf{w}^2)} \right)$.  
 Using Lemma 1, we have $L_2$ convergence for the estimator $\hat{D}_h$ and, via Lemma 2, this implies convergence in probability.
 Therefore $\exists M \in \mathbb{N}$ such that $\forall n \ge M$ 
 $$ \mathbb{P}\left( \left|\hat{D}_h \left(X_{1:n} || Y_{1:n}^{(\mathbf{w})} \right) - D_h \left(p || q^{(\mathbf{w})} \right) \right| > \delta \right) < \epsilon. $$
 Therefore $\forall n \ge M$ $$ \mathbb{P} \left( \argmax_{\mathbf{w} \in E} \hat{D}_h \left( X_{1:n} || Y_{1:n}^{(\mathbf{w})} \right) = \mathbf{w}^* \right) > 1-\epsilon .$$
\end{proof}

If we make the (generally unrealistic) assumption that the space of possible parameters, $S$, is finite, then we are able to show that, in the limit of a large number of samples, we can recover the optimum parameters with probability 1.
We next explore how this can be extended to a compact, continuous space of parameters, $S$, provided we make assumptions about the structure of dependence of the distance estimator on parameters $\mathbf{w} \in E$.
We show that assumptions about this dependence structure of the estimator on the parameters can be justified by considering some of the details of Algorithm 2. 

\begin{lemma} \label{lemma:piecewise_const}
Consisder the estimator of the Hellinger distance as a function of parameters, $\mathbf{w}$, such that $$L(\mathbf{w}) = \hat{D}_h \left(X_{1:n} || Y_{1:n}^{(\mathbf{w})}\right). $$
Then $L(\mathbf{w})$ is piecewise constant with respect to parameters $\mathbf{w} \in E$, with finitely many discontinuities.
\end{lemma}

\begin{proof}
Suppose that, according to Algorithm 2, at generation $t$, we have generated pseudo data $\{\mathbf{x}^{i*}\}_{i=1}^N$,
which we summarise via summary statistics $s(\mathbf{x}^{i*})=\left( s_1, \ldots , s_{\kappa}\right) \in \mathbb{R}^{\kappa}$ and $\kappa$ is the number of summary statistics used to summarise the model output.
The parameters $\mathbf{w} \in \mathbb{R}^{\kappa}$ are summary statistic weights, and these are used within a weighted Euclidan distance function
\begin{align*}
  d_{\mathbf{w}}\left(s(\mathbf{x}),s(\mathbf{y})\right) &= \sum_{i=1}^\kappa w_i^2 (s_i - s_{\text{obs}_i})^2 \\
  &= (s(\mathbf{x})-s(\mathbf{y}))^\top \Sigma_w^\top \Sigma_w (s(\mathbf{x})-s(\mathbf{y})) \\
  &= \left(\Sigma_w s(\mathbf{x})- \Sigma_w s(\mathbf{y})\right)^\top \left( \Sigma_w s(\mathbf{x}) - \Sigma_w s(\mathbf{y}) \right),  
\end{align*}
where $\Sigma_w = \text{diag}(\mathbf{w})$, to compare the pseudo data with observed data.
We note that this is equivalent to stretching the space in which the pseudo data lies via the matrix $\Sigma_w$, and using the usual Euclidean distance.

Consider a small perturbation in parameter space $\mathbf{w}=\mathbf{w}_0+\epsilon$, with $||\epsilon|| \ll 1$. 
Then
$$
\Sigma_{\mathbf{w}} = \text{diag}(\mathbf{w}_0+\epsilon) = \text{diag}(\mathbf{w}_0) + \text{diag}(\epsilon) = \Sigma_{\mathbf{w}_0} + \Sigma_{\epsilon}.
$$
Using this decomposition of $\Sigma_\mathbf{w}$ gives, for the weighted Euclidean distance,
\begin{align*}
 d_{\mathbf{w}_0+\epsilon}(s(\mathbf{x}),s(\mathbf{y})) &= \left(\left(\Sigma_{\mathbf{w}_0} + \Sigma_{\epsilon} \right) s(\mathbf{x}) - \left(\Sigma_{\mathbf{w}_0} + \Sigma_{\epsilon} \right) s(\mathbf{y})\right)^\top \\
&\, \, \, \, \, \, \, \left(\left(\Sigma_{\mathbf{w}_0} + \Sigma_{\epsilon} \right) s(\mathbf{x}) - \left(\Sigma_{\mathbf{w}_0} + \Sigma_{\epsilon} \right) s(\mathbf{y})\right) \\
&= \left( \Sigma_{\mathbf{w}_0} s(\mathbf{x})- \Sigma_{\mathbf{w}_0} s(\mathbf{y}) \right)^\top \left( \Sigma_{\mathbf{w}_0} s(\mathbf{x})- \Sigma_{\mathbf{w}_0} s(\mathbf{y}) \right) \\
&\, \, \, \, \, \, \, + \left( \Sigma_{\epsilon} s(\mathbf{x})- \Sigma_{\epsilon} s(\mathbf{y}) \right)^\top \left( \Sigma_{\mathbf{w}_0} s(\mathbf{x})- \Sigma_{\mathbf{w}_0} s(\mathbf{y}) \right) 
\\ &\, \, \, \, \, \, \, + \left( \Sigma_{\mathbf{w}_0} s(\mathbf{x})- \Sigma_{\mathbf{w}_0} s(\mathbf{y}) \right)^\top \left( \Sigma_{\epsilon} s(\mathbf{x})- \Sigma_{\epsilon} y \right)
+ \mathcal{O}(\epsilon^2) \\
&= d_{\mathbf{w}_0}(s(\mathbf{x}),s(\mathbf{y})) + A + A^\top + \mathcal{O}(\epsilon^2),
\end{align*}
where $A = \left( \Sigma_{\mathbf{w}_0} s(\mathbf{x})- \Sigma_{\mathbf{w}_0} s(\mathbf{y}) \right)^\top \left( \Sigma_{\epsilon} s(\mathbf{x})- \Sigma_{\epsilon} s(\mathbf{y}) \right) $
which is linear in $\epsilon$.
Suppose we order the pseudo-data such that $\mathbf{x}^{j*}$ is the $j$th closest point to the observed data based on $d_{\mathbf{w}_0}$.
Provided that 
\begin{equation} \label{jump_condition}
d_{\mathbf{w}_0}(s(\mathbf{x}^{n*}),s(\mathbf{y})) - d_{\mathbf{w}_0}(s(\mathbf{x}^{(n+1)*}),s(\mathbf{y})) > A + A^\top + \mathcal{O}(\epsilon^2),
\end{equation}
then making this perturbation in $\mathbf{w}$ will not change which parameter samples are selected, as the same pseudo data will remain closest to the observed data, $\mathbf{y}$.
If the same parameter samples are selected, then the value of $L(\mathbf{w})$ will remain constant as a function of $\mathbf{w}$ under the perturbation $\mathbf{w}=\mathbf{w}_0+\epsilon$.

In cases where \eqref{jump_condition} does not hold, there will be a jump discontinuity in $L(\mathbf{w})$ as different parameter samples are selected.
This will occur finitely many times corresponding to the finite number, $N$, of points in the sample of pseudo data, $\{\mathbf{x}^{i*}\}_{i=1}^N$.
\end{proof}

We confirm computationally that $L(\mathbf{w})$ is piecewise constant for the test problem described in Section \ref{toy_model}, and show this in Figure \ref{fig:toy_model_piecewise_const}.  
\begin{figure*}[h!]
\begin{center}
\includegraphics[width=\textwidth]{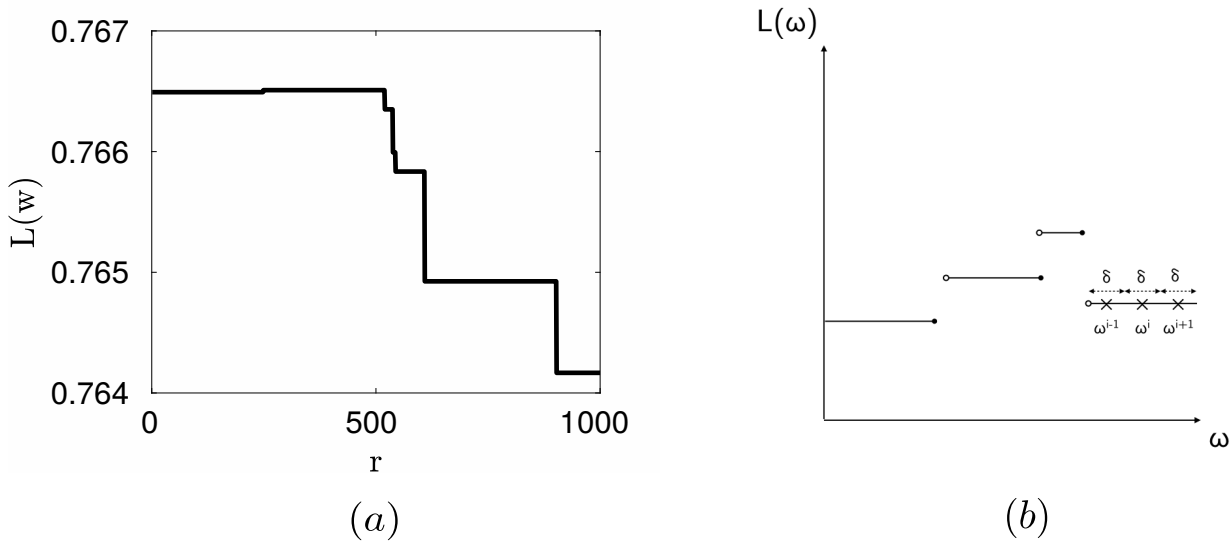}
\end{center}
\caption{The estimator $L(\mathbf{w}) = \hat{D}_h \left(X_{1:n} || Y_{1:n}^{(\mathbf{w})}\right)$ is piecewise constant for the toy model described in Section \ref{toy_model} when $X_{1:n}$ are samples from the prior and $Y_{1:n}^{(\mathbf{w})}$ are samples from the approximate posterior generated via ABC with summary statistics weights $\mathbf{w}$.
After a single generation of ABC-SMC, we optimize $L(\mathbf{w})$ as a function of $\mathbf{\mathbf{w}}$ to find a local maximum $\mathbf{w}^*$.
In (a), we then consider the value of $L(\mathbf{w})$ on a line of parameters in parameter space $\mathbf{w}=\mathbf{w}^* + 10^{-4} r \vec{\eta}$, where $\vec{\eta} \sim N(\mathbf{0},I_{\kappa})$ is a random choice of direction, $I_{\kappa}$ is the $\kappa \times \kappa$ identity matrix and $r$ parameterises a line in this direction.
$L(\mathbf{w})$ is piecewise constant in $\mathbf{w}$ as shown in Lemma \ref{lemma:piecewise_const}.
For a piecewise constant function $L(\mathbf{w})$, we can choose $\delta$ and $\mathbf{w}^j$ for $j \in \{1, \dots, \kappa \}$ such that $L(\mathbf{w})$ is locally constant.
We illustrate such a choice of $\delta$ and $\mathbf{w}^j$ in (b).
}
\label{fig:toy_model_piecewise_const}
\end{figure*}


\begin{lemma}
 Assume $S \subset \mathbb{R}^s$ is compact and that a unique $\mathbf{w}^* \in E$ maximises $D_h\left(p || q^{(\mathbf{w})} \right)$.
 Assume the conditions of Lemma 1 hold for distributions $p$ and $q^{(\mathbf{w})}$.
 Suppose $L(\mathbf{w})$ is piecewise constant in $\mathbf{w}$ with finitely many jump discontinuities. Then
 $$ \lim_{n \rightarrow \infty} \mathbb{P} \left( \argmax_{\mathbf{w} \in E} \hat{D}_h \left( X_{1:n} || Y_{1:n}^{(\mathbf{w})} \right) = \mathbf{w}^* \right) = 1. $$
\end{lemma}
\begin{proof}
We can choose a $\delta >0$ and finitely many $\mathbf{w}^j$, $j \in \{1, \dots, \chi \}$ such that every point $\mathbf{w}$ is within a ball of radius $\delta$ from some $\mathbf{w}^j$ (see Figure \ref{fig:toy_model_piecewise_const} and note that $L$ is locally constant).
Then $\forall \mathbf{w} \in E \, \, ||\mathbf{w}-\mathbf{w}^j|| < \delta \, \implies \, L(\mathbf{w}) = L(\mathbf{w}^j) = L_j. $
Since there are finitely many values $\mathbf{w}^j$ corresponding to distinct unique values $L^j$,  we can apply the result from Theorem 3 to give the required result.
\end{proof}

\begin{theorem}
Suppose that $k\ge 2$ and that $\mathcal{M} = \text{supp}(p)$.
Assume that (a) $q$ is bounded above, (b) $p$ is bounded away from zero, (c) $p$ is uniformly Lebesgue approximable,
(d) $\exists \, \delta_0$ such that $\forall \delta \in (0,\delta_0) \, \int_{\mathcal{M}} H(x,p,\delta,1/2)p(x) \mathrm{d}x < \infty,$
(e) $\int_{\mathcal{M}}||x-y||^{\gamma} p(y) \mathrm{d}y < \infty$ for almost all $x \in \mathcal{M}$, $\int \int_{\mathcal{M}^2}||x-y||^{\gamma} p(y) p(x) \mathrm{d}y \mathrm{d} x < \infty$.
Assume $S \subset \mathbb{R}^s$ is compact and that a unique $\mathbf{w}^* \in E$ maximises $D_h\left(p || q^{(\mathbf{w})} \right)$.

Then
\begin{equation}
  \lim_{n \rightarrow \infty} \mathbb{P} \left( \argmax_{w \in E} \hat{D}_h \left( X_{1:n} || Y_{1:n}^{(\mathbf{w})} \right) = \mathbf{w}^* \right) = 1
\end{equation}
\end{theorem}

\begin{proof}
 Apply the Lemma 4 to show that $L(\mathbf{w})$ is piecewise constant in $\mathbf{w}$. Apply Lemma 5 to give the desired result.
\end{proof}

To summarise, with this nearest neighbour estimator, under conditions on $p$ and $q$, we have $L_2$ convergence of the estimator and this ensures that, in the limit of more samples,
estimates of the distance between $p$ and $q$ will become more concentrated around the true distance, such that optimising the estimate of the distance will give the true optimum, $\mathbf{w}^*$, by making use of the piecewise constant structure of $L(\mathbf{w})$.
We assume in Theorem 2 (and elsewhere) that the space of parameters, $S$, is compact. In practice this is not a problem, since we can work with a constrained optimization problem and assume that the summary statistic weights lie within a large but finite region.
Although there are several conditions on the prior distribution, $p$, most reasonable choices of prior distribution will satisfy these, and only a single condition on the approximate posterior distribution, $q$, is assumed.
In the limit of having more samples from the distributions $p$ and $q^{(\mathbf{w})}$, selecting summary statistic weights, $\mathbf{w}$, based on optimizing the estimate from $\hat{D}_h$ will converge to give the true optimum, $\mathbf{w}^*$, of this distance between distributions.
For fixed weights $\mathbf{w}^*$ in the ABC distance function, ABC-SMC will target the correct posterior distribution.

\section{Examples} \label{Examples}
We apply our algorithm of automatic, adaptive weighting of summary statistics to a variety of test problems, including a toy model and several problems based on different chemical reaction networks. 
The dynamics of these networks are simulated stochastically using Gillespie's direct method \cite{gillespie1977exact},
which allows us to sample trajectories directly from the model.
Although for some of these models it is possible to solve for the likelihood analytically, we attempt parameter inference by simulation,
since solving for the likelihood is very computationally expensive and, in general, the analytical solution is not available. 
The summary statistics collected for each of the chemical reaction network problems are in the form of a time series, to imitate data that could be collected from a biological experiment.

To demonstrate the effectiveness of taking a flexible choice of distance weights, we make two comparisons.
Firstly, we compare results obtained using Algorithm 2 to those generated using a uniform choice of weights:
$w_i = 1 \hspace{2mm} \forall i$. 
Secondly, we compare to results generated using weights that scale with each summary statistic.
Here we use $w_i = 1/\sigma_i \hspace{2mm} \forall i$, where $\sigma_i$ is the standard deviation of the given summary statistic.
This is a frequently used choice of weight for summary statistics \cite{beaumont2002approximate}. 
We note that a table summarising the parameters used in the implementation of all the test problems can be found in Appendix A.

\subsection{Toy model} \label{toy_model}

We consider a tractable toy problem with a sufficient statistic to illustrate our method.
We observe 
\begin{equation}
 x_i \sim \text{Unif}([0,\theta]),
\end{equation}
for $i=1,\dots, r$.
In this case, the maximum of the observed values is a sufficient statistic, $s(\mathbf{x}) = \max_i x_i$.
We can sample the true posterior distribution $p(\theta | \mathbf{y})$ directly via MCMC and compare to the approximations obtained via ABC.

\begin{figure*}[h!]
\begin{center}
\includegraphics[width=0.45\textwidth]{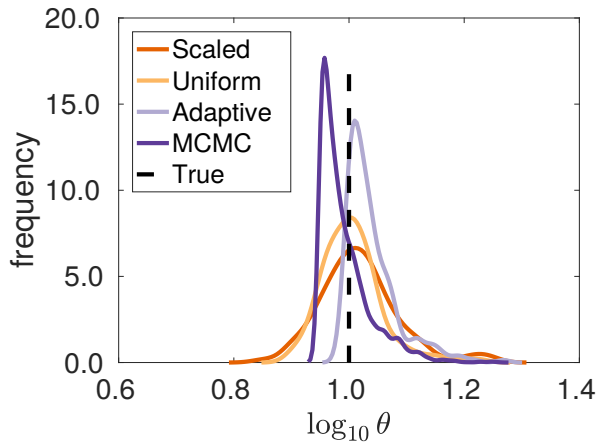}
\end{center}
\caption{Posterior for parameters $\theta$ of the uniform toy model for different weights in the ABC distance function. 
ABC-SMC was used to provide estimates of the posterior, with five generations and $N=50,000$ simulations at each generation with the posterior constructed from the closest $5\%$ of the simulations ($\alpha = 0.005$).
}
\label{fig:toy_model}
\end{figure*}

The results in Figure \ref{fig:toy_model} indicate that the method in Algorithm 2 is able to produce a higher quality approximation of the posterior for a given number of parameter samples compared to other methods of weighting the summary statistics.
The true parameter is $\theta = 10$, a prior uniform on the logarithm of the parameters over the interval $[10^{0},10^2]$ was used, and
$r=10$ samples of the uniform model were used as the dataset.

\subsection{Death process} \label{exp-decay}

For our first test problem, we consider estimating the rate parameter for a single, first order degradation reaction:
\begin{eqnarray}
\label{exponential}
\ce{A
->[k]
{ \emptyset }.
} 
\end{eqnarray}
We will consider for this, and subsequent, test problems that time has been non-dimensionalized.
Initially, we assume there are $A(0) = 10$ particles in the system, which is observed over a (non-dimensional) time period $[0,20]$.
We assume it is possible to measure the state of the system (in this case the number of molecules of species $A$) without observation noise
at given time points $t_0, t_1, \ldots , t_n$. 
For this test problem, we assume that we measure at $n$ equally spaced time intervals, where $n=32$. 

As our summary statistics, we take $s(\mathbf{x}) = [A(t_0), A(t_1), \linebreak \ldots , A(t_n), z] $
where $z$ is an observation of a random variable $Z \sim N(0,\sigma^2)$ that is uncorrelated with the death process.
We suppose that the scale of the variance, $\sigma$, is different to the scale of the observations of the exponential decay process, 
giving a simple system with a two-dimensional parameter to infer: $\theta = (k, \sigma)$.
Note that the scale of $z$ is determined by the standard deviation, $\sigma$, but the scale of the death process is affected by the initial condition, $A(0) = 10$,
resulting in two distinct scales in these summary statistics.

\begin{figure*}[h!]
\begin{center}
\includegraphics[width=\textwidth]{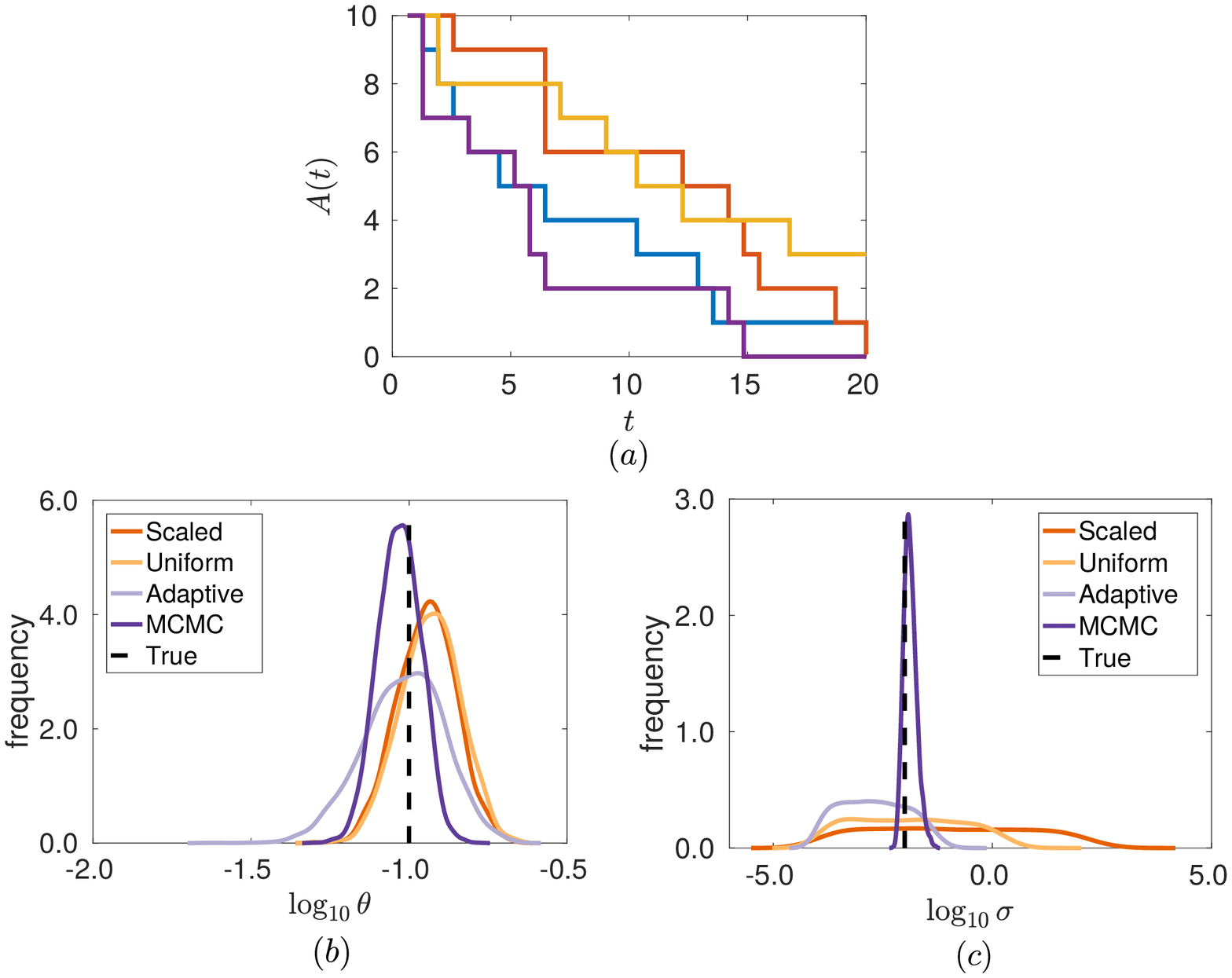}
\end{center}
\caption{Posteriors for parameters $k$ and $\sigma$ in the death process test problem for different weights in the ABC distance function. 
ABC-SMC was used to provide estimates of the posterior, with five generations and $N=50,000$ simulations at each generation with the posterior constructed from the closest $5\%$ of the simulations ($\alpha = 0.05$).
(a) shows typical output from the model for the true parameters.
The posteriors for $k$ are given in (b) and for $\sigma$ in (c).
}
\label{fig:Uniform}
\end{figure*}

Results of parameter inference for this system using ABC-SMC are shown in Figure \ref{fig:Uniform}, where the true parameters used are $\theta = (0.1, 0.01)$ and a prior uniform on the logarithm of each of the parameters over the interval $[10^{-3},10^3]$ was used.
Here we show marginal posterior distributions generated using uniform weights, weights scaled with the standard deviation of each summary statistic,
and adaptively chosen weights via the method outlined in Algorithm 2.
We observe similar performance in identification of the decay parameter $k$ using uniform weights, scaled weights and the adaptive choice of weights. 
Scaling the summary statistics with their standard deviation results in a posterior that does not provide much information over the prior for $\sigma$, since all the summary statistics are assumed to be equally informative which is not the case here.
Note that only one summary statistic provides information about the random variable $Z$, whereas the other $n+1$ summary statistics (which are observations of the decay process at time points $\{t_i\}_{i=0}^n$) provide information about the decay of species $A$.
The summary statistic weights chosen via the search process outlined in Algorithm 2 give rise to a posterior that outperforms the posteriors generated using uniform weights and scaled weights for the second parameter $\sigma$,
since only a single summary statistic provides relevant information for this parameter.

\subsection{Dimerization system} \label{Dimerization}

To examine a system with multiple scales, we consider also a dimerization system,
which undergoes a fast initial transient followed by slower subsequent dynamics \cite{lester2015adaptive}.
The dimerization system consists of the following reactions:
\begin{eqnarray*}
\label{R1}
R_1&:& ~~ \ce{S_1
->[k_1]
{ \emptyset },
}; \\
\label{R2}
R_2&:& ~~ \ce{S_2
->[k_2]
{ S_3 },
}; \\
\label{R3}
R_3&:& ~~ \ce{S_1 + S_1
->[k_3]
{ S_2 },
}; \\
\label{R4}
R_4&:& ~~ \ce{S_2
->[k_4]
{ S_1 + S_1 }.
}
\end{eqnarray*}

We take initial conditions $S_1(0)=10^5$, $S_2(0)=0$, $S_3(0)=0$ and consider an observational time period of $[0,100]$ with $n=32$ geometrically spaced observations (to capture the multiple timescales present), without observational noise.
For the dimerization system, we take the time series $s(\mathbf{x}) = [S_1(t_0), \ldots , S_1(t_n),S_2(t_0), \ldots , \linebreak S_2(t_n), S_3(t_0), \ldots , S_3(t_n)]$
as summary statistics and infer the four-dimensional parameter $\theta = (k_1, k_2, k_3, k_4)$.
\begin{figure*}[h!]
\begin{center}
\includegraphics[width=0.99\textwidth]{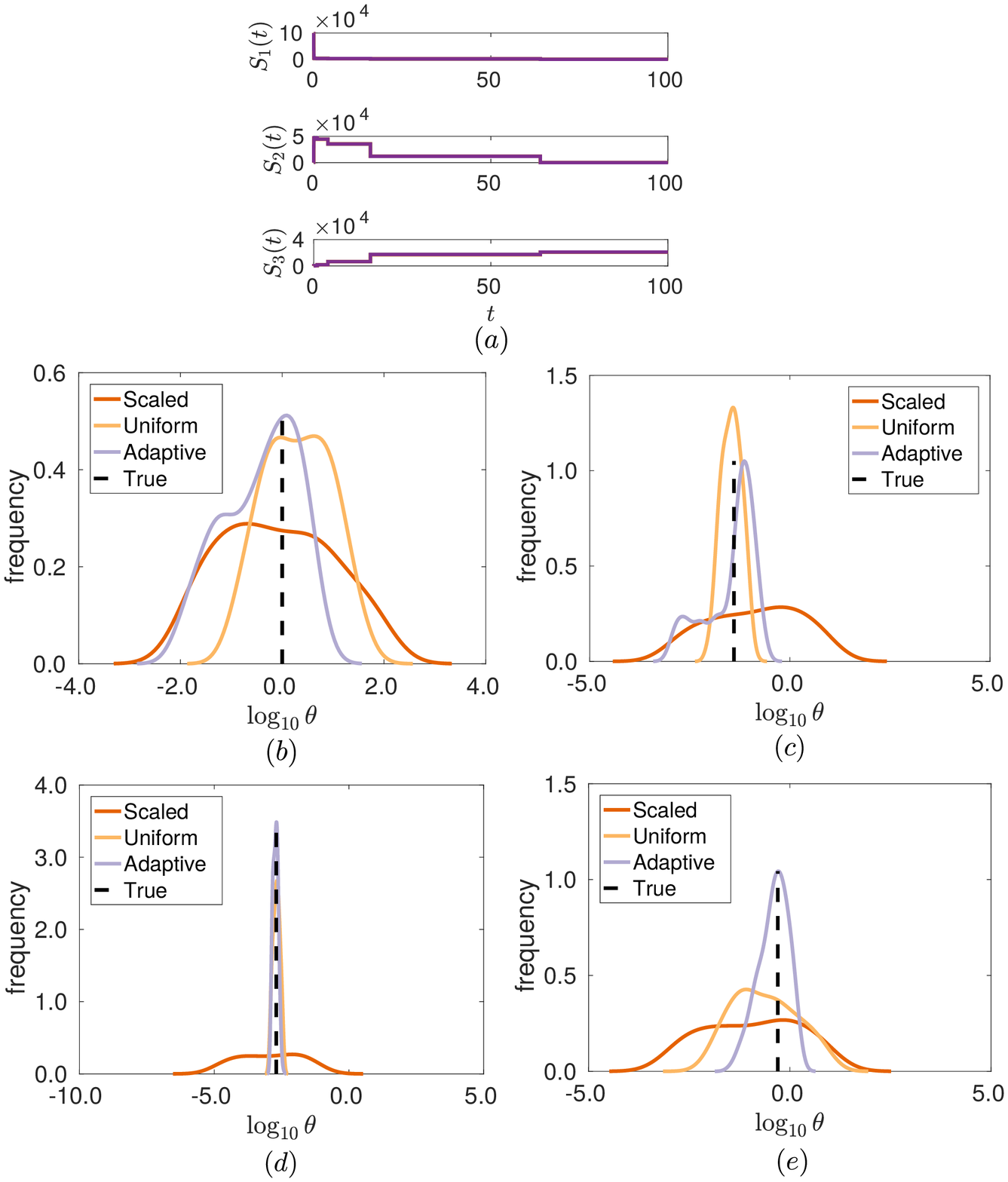}
\end{center}
\caption{Posteriors for parameters $\theta = (k_1, k_2, k_3, k_4)$ in the dimerization system for different weights in the ABC distance function. 
ABC-SMC was used with five generations and $N=50,000$ simulations at each generation with the posterior constructed from the closest $5\%$ of the simulations ($\alpha = 0.05$).
(a) shows typical output from the model for the true parameters, for each species, $S_i$.
Posterior marginal distributions for parameters $k_1,k_2,k_3,k_4$ are shown in (b) to (e).
}
\label{fig:dimerization}
\end{figure*}
\noindent We note that for a choice of parameter $\theta^* = (1, 0.04, \linebreak 0.002, 0.5)$, and the given initial conditions,
we obtain a fast decay of species $S_1$ and accumulation of species $S_2$,
followed by a slower decay of $S_2$ and accumulation of $S_3$ (see Figure \ref{fig:dimerization}(a)).

The results of parameter inference for this system can be seen in Figure \ref{fig:dimerization}.
The true parameters used are $\theta = (1, 0.04, 0.002, 0.5)$, and we apply a prior uniform on the logarithm of the parameters over the intervals
$[10^{-2},10^2]$, $[10^{-3},10^1]$, $[10^{-5},10^{-1}]$, $[10^{-3},10^1]$, respectively, for each parameter.
Parameters $k_1$ and $k_2$ are clearly identified by the adaptive choice of weights. 
The fast transient behaviour initially involves reactions at rate $k_1$, while $k_2$ corresponds to the longer timescale accumulation of species $S_3$.
Parameters $k_3$ and $k_4$ are harder to identify with broader resulting posteriors, but again the adaptive algorithm does a better job at excluding regions of search space than a uniform choice of weights, or a scaling with the standard deviation.
Scaling by the standard deviation is a poor choice here because for some of the time points, particularly in the fast initial transient region, there is no variation between the synthetic datasets.  

\subsection{Simple spatial model}
Spatial models produce very high dimensional data, containing information about dynamics in both space and time.
Here, we consider a simple spatial model in one dimension to describe the spreading of particles by diffusion without volume exclusion. 
We divide our spatial domain $X \in [-1, 1]$ into $m$ boxes or voxels, and label the numbers of particles in voxels $1, \ldots , m$ as $S_1, \ldots , S_m$, respectively.
Particles can jump between neighbouring voxels at rate $\theta = D/h^2$, where $D$ is the macroscopic diffusion constant and $h$ is the width of the voxel.
We assume zero flux conditions at $X=\pm 1$ and take $m=8$, so that $h=1/4$.
As an initial condition, we place 10 particles in each of the $m/2$ voxels on the left-hand side of the domain where $x<0$,
and allow the system to evolve over the time interval $[0,20]$.
We observe the system at $n=8$ equally spaced time points, and take as our summary statistic the time series for each voxel, $s(\mathbf{x}) = [S_1(t_0), \ldots , S_1(t_n),S_2(t_0), \ldots , S_2(t_n), \linebreak \ldots ,S_m(t_0), \ldots , S_m(t_n)]$,
where $S_i(t_j)$ is the number of particles in voxel $i$ at time point $t_j$.
Using synthetic data simulated with $\theta=0.1$, we attempt to recover the jump rate $\theta$.
The results of parameter inference for this problem are shown in Figure \ref{fig:spatial_test_problem}, where we have used a prior uniform on $\log_{10}(\theta)$ over the interval $[10^{-4}, 10^0]$.
We successfully obtain an informative unbiased posterior for $\theta$ using the adaptive choice of weights, with a notable improvement in comparison to the other methods for selecting the weights.

\begin{figure*}[h!]
\begin{center}
\includegraphics[width=\textwidth]{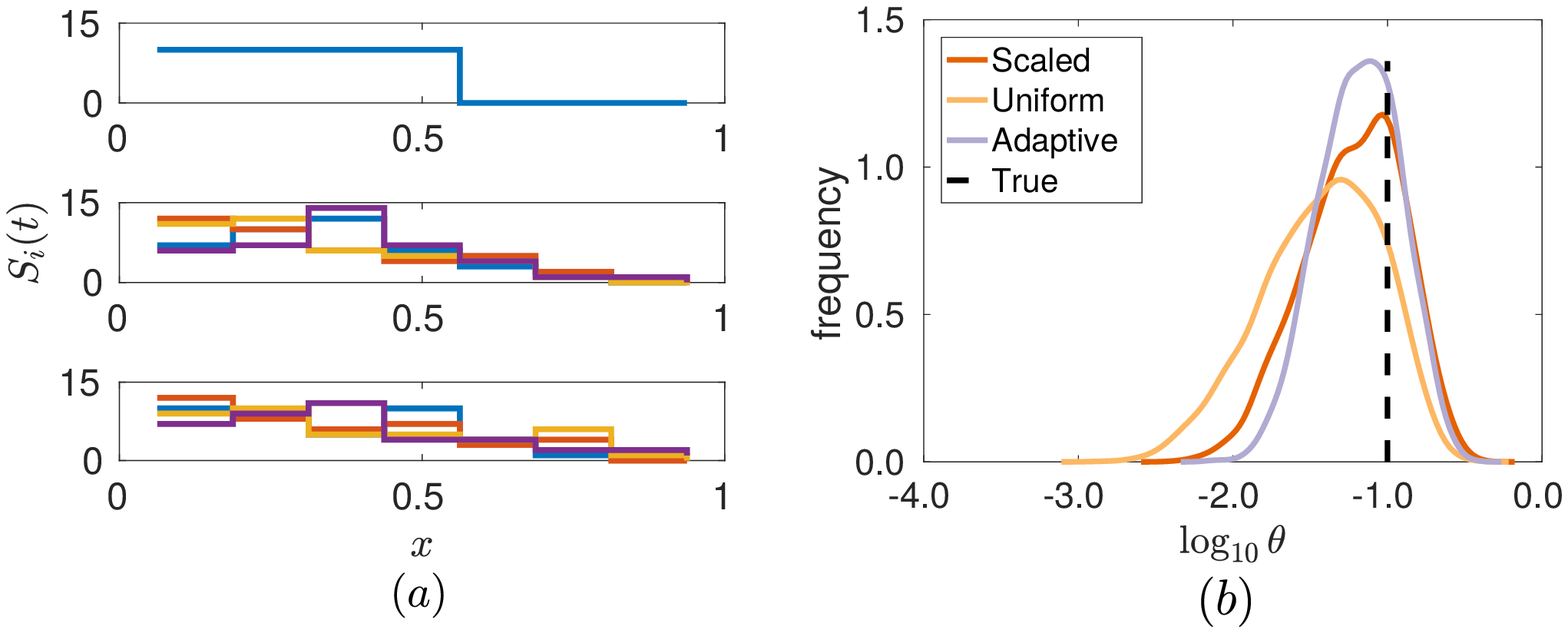}
\end{center}
\caption{Posteriors for parameter $\theta$ in the simple diffusion model for different weights in the ABC distance function. 
ABC-SMC was used for the inference with five generations and $N=50,000$ simulations at each generation with the posterior constructed from the closest $5\%$ of the simulations ($\alpha = 0.05$).
(a) shows the spatial profile at three different time points ($t=0,10,20$) and demonstrates the variability in the output for this spatial process across four realizations with the same parameter, $\theta =0.1$. 
In (b), we compare the posteriors obtained for $\theta$ with different choices of weights.
}
\label{fig:spatial_test_problem}
\end{figure*}

\subsection{Computational overhead} \label{Computational_overhead}
If our proposed approach of adapting the weights of each of the summary statistics is to be used in practice, we must ensure that the increases in the quality of the resulting posterior
justify the computational overhead required for the search process.
Otherwise, it would be preferable simply to generate the posterior using ABC-SMC with more samples.
Therefore we are interested in evaluating the computational overhead of the search process,
and how to limit the cost of the search in higher dimensions.

Using the dimerization test problem, as described in Section \ref{Dimerization}, we ran Algorithm 2 with $N_1=5,000$, $\alpha_1 = 5\%$.
To compare this to ABC-SMC with uniform weights, we performed parameter inference with uniform weights using both $N_1=5,000$, $\alpha_1 = 5\%$ and $N_2=5,600$, $\alpha_2=4.46\%$.  
The value of $N_2$ was chosen such that an equal length of computation time was spent in the search steps to find the summary statistic weights in Algorithm 2, 
as was spent in generating extra samples in ABC-SMC with uniform weights.
A corresponding lower value of $\alpha$ was chosen so that the number of particles in the parameter sample was equivalent.

In this case, adaptively choosing weights using Algorithm 2 resulted in a significantly greater distance between the prior and posterior, 
and reduced the bias in the posterior compared to running ABC-SMC with more samples, as measured by the distance between the maximum posterior estimate and the true parameters.
These results, which represent improvements in the posterior for the same computational cost, are shown in Table 1 and the same procedure was used for the other test problems.

\begin{table*}[t]
  \centering
    \begin{tabular}{ | l | l | l |}
    \hline
    \textbf{Test problem} & \textbf{Hellinger distance} & \textbf{Bias in posterior} \\
    & \textbf{between prior and posterior} & \\ \hline
    Toy model & 0.792/0.783/\textbf{0.803} & 0.047/0.047/\textbf{0.032} \\ \hline
    Death process & 0.838 / 0.825 / \textbf{0.853} & 0.136 / \textbf{0.114} / 0.260 \\ \hline
    Dimerization & 0.923 / 0.923 / \textbf{0.937} & 0.125 / 0.243 / \textbf{0.057} \\ \hline
    Diffusion & 0.723 / 0.730 / \textbf{0.771} & 0.486 / 0.491 / \textbf{0.130} \\ \hline
    \end{tabular}
        \caption{Performance of Algorithm 2 compared with increasing the number of samples in ABC-SMC.
        Results are shown for each of the test problems in the form: ABC-SMC with $N_1$ and $\alpha_1$ / ABC-SMC with $N_2$ and $\alpha_2$ / Algorithm 2 with $N_1$ and $\alpha_1$.
        Highlighted in bold is the method with best performance according to each metric.}
\end{table*}

\subsection{Consistent weights} \label{Consistent}
Ideally, our search process should find the global optimum weight vector, so that if Algorithm 2 is run multiple times the same weight vector is obtained.
In practice, for the examples we have explored, the function to be optimized (distance between prior and posterior as a function of the distance weights) is very flat with respect to some of the distance weights.
This makes it hard to consistently identify a global maximum.
In Figure \ref{Fig:consistent_weights}, we explore how the chosen weights vary for the toy model and the death process examples.
We can interpret this as the algorithm identifying the informative summary statistics and appropriately using the information from these, 
while allowing weights for other summary statistics to take a range of values without much effect on the resulting posterior.
The largest weight is given to the most informative summary statistic.

The weights found for different runs of the algorithm are highly correlated, however, as expected.
To better compare the weights found by optimization across runs of the algorithm, we subtract the mean of the weights for each run of the Algorithm 2. 
This highlights the summary statistic $z$ for the death process test problem as highly informative (see Figure \ref{Fig:consistent_weights}(b)), which agrees with our intuition,
since only this summary statistic gives informative about the parameter $\sigma$, whereas any of the others can provide information about the decay parameter, $k$.

\begin{figure*}[h!]
\begin{center}
\includegraphics[width=\textwidth]{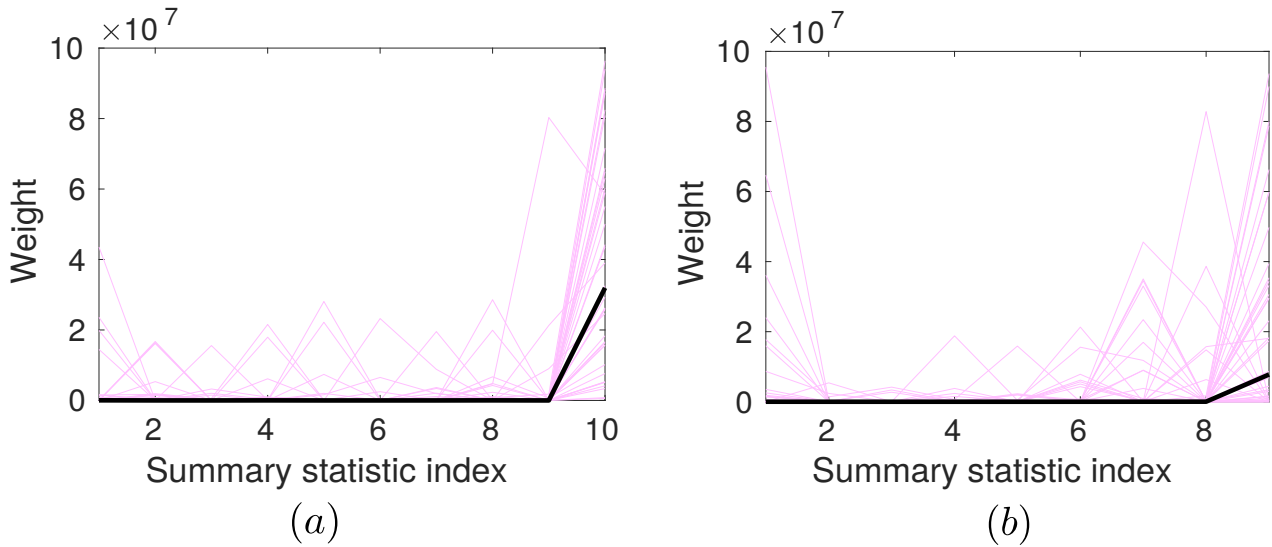}
\end{center}
\caption{The optimal distance weights found from the search procedure after 40 successive runs of Algorithm 2
on the toy model test problem (described in Section \ref{toy_model} in (a) and on the death process test problem (described in Section \ref{exp-decay}) in (b).
Parameters as for Figures \ref{fig:toy_model} and \ref{fig:Uniform} in each case.
The faint purple lines show the resulting summary statistics weights from repeated runs of the ABC distance weight algorithm, while the black line shows the mean of the weights selected.
}
\label{Fig:consistent_weights}
\end{figure*}

\section{Discussion}
\label{Summary} 
In this work, we have presented a method for improving the quality of posteriors resulting from approximate inference using ABC-SMC by optimizing the weights of the ABC distance function, $d_{\mathbf{w}}(s(\mathbf{x_1}),s(\mathbf{x_2}))$.
By applying the methodology to several test problems we have demonstrated that our novel, adaptive method allows effective combination of summary statistics. 
We see superior performance using our algorithm in comparison with naive choices of uniform weights or using the scale of the summary statistics.
Further benefits of adapting the weights include removing the requirement for design and selection of summary statistics `by hand'. 

\subsection{Comparison to dimensionality reduction methods} \label{DimReduction}
Adaptively choosing the summary statistic weights within the ABC distance function can be seen as achieving a similar goal to summary statistic dimension reduction techniques \cite{nunes2010optimal,blum2013comparative}.
These techniques either project high-dimensional summary statistics into a lower dimensional subspace, or select an optimal subset of summary statistics via some optimality criterion.
In contrast, a similar effect is achieved here when the statistics are combined in the weighted Euclidean distance function, $d_{\mathbf{w}}(\mathbf{x_1},\mathbf{x_2})$,
by weighting summary statistics to take account of both their inherent scale,
and also their relative contribution towards the posterior distribution.
Uninformative summary statistics are automatically assigned a lower weighting, while more informative summary statistics are given high weights relative to their scale.
 
\begin{table*}[t]
  \centering
    \begin{tabular}{ | l | l | l |}
    \hline
    \textbf{Test problem} & \textbf{Hellinger distance}& \textbf{Bias in posterior} \\
    & \textbf{between prior and posterior} & \\ \hline
    Toy model &     \textbf{0.800} / 0.793 / 0.790 & \textbf{0.045} / 0.051 / 0.048 \\ \hline
    Death process & \textbf{0.914} / 0.0896 / 0.844 & \textbf{1.078} / 1.304 / 2.641 \\ \hline
    Dimerization & \textbf{0.877} / 0.858 / 0.876 & 1.645 / 1.229 / \textbf{0.715} \\ \hline 
    Diffusion & \textbf{0.737} / 0.673/ 0.721 & 0.451 / 0.511 / \textbf{0.275} \\ \hline
    \end{tabular}
        \caption{Comparison of the quality of the posteriors obtained using different methods to combine summary statistics.
        Results given as adaptive method/\citet{barnes2012considerate}/\citet{fearnhead2012constructing}.
        Bold text highlights the best performance on a metric for a test problem.}
\end{table*}

Previous subset selection methods have used criteria for approximate sufficiency of a subset of statistics to test whether adding a new
statistic results in a change in the posterior above a certain threshold \cite{joyce2008approximately}; 
minimising an information criterion based on knn-entropy over all subsets of summary statistics \cite{nunes2010optimal}; 
and reducing loss of information by adding summary statistics until the KL divergence between the resulting posteriors is below a threshold \cite{barnes2012considerate}.
All of these methods seek to choose a lower dimensional subset of a given list of summary statistics. 
Using this lower dimensional subset increases the acceptance rate for samples in ABC by avoiding the curse of dimensionality for the data.
However, the results depend on the order in which the summary statistics (or subsets) are analysed.

A popular method, implemented in packages such as abctools \cite{nunes2015abctools}, is the semi-automatic ABC approach of Fearnhead and Prangle \cite{fearnhead2012constructing}. 
This approach uses a projection method to find informative linear combinations of statistics by fitting a regression for each parameter in the model. 
The result is a reduction from the original high-dimensional set of summary statistics to a new lower dimensional set of summary statistics with the same dimensionality as the parameter space.
Improved results are seen by using a pilot run of ABC to choose a subset of parameter space as a training region for the regression.
Further improvements are obtained by extending the vector of summary statistics by concatenating with a non-linear transformation of the same summary statistics, 
$s(\mathbf{x}) = (s,s^2, s^3, s^4)$, where $s$ is a given vector of summary statistics and the superscripts indicating raising these to the given power.
This method of Fearnhead and Prangle \cite{fearnhead2012constructing} uses contributions from all of the summary statistics and should optimize the mean quadratic loss.

We tested our adaptive weight selection algorithm against the semi-automatic ABC method \cite{fearnhead2012constructing}, and the subset selection method of \citet{barnes2012considerate} based on an approximate sufficiency criterion.
In general, for the test problems considered, as described in Section \ref{Examples}, our method outperforms the competing methods, as shown by the metrics in Table 2.
A larger value of the Hellinger distance indicates a greater distance between prior and posterior.
The bias gives the distance between the posterior mean and the true parameter value.
In implementing these methods, we have used only ABC rejection sampling, equivalent to a single generation of ABC-SMC, to compare the methods.
In practice, these results mean that our method outlined in Algorithm 2 for adaptively choosing the weights of summary statistics 
produces a more informative posterior than competing methods based on dimensionality reduction of summary statistics.

\subsection{Further work}
Our method for automatically adapting the weights of the ABC distance function could be combined with other methods for dimensionality reduction of summary statistics to further improve the quality of posteriors produced with ABC for given computational effort.
A particular area to consider would be how best to combine optimization of the distance weights for ABC and dimensionality reduction of the summary statistics.
These are related approaches that can work well together. 
One approach that could be explored, for example, is enforcing some sparcity of the weights during the search step of the weights optimization.
By setting some weights to be explicitly zero, we exclude the corresponding summary statistics, effectively reducing the dimensionality of our summary statistics.
Further investigations could explore how best to sample sparse subsets of weights in high dimensions.

\subsection{Conclusion}
In summary, we propose a computationally efficient search procedure to identify a set of optimum weights to allow us
to combine summary statistics within the ABC distance function in such a way that the gain in information in the posterior over the prior is maximized.

\begin{acknowledgements}
We thank James Martin and Geoff Nichols for helpful discussions of this work.
This work was supported by funding from the Engineering and Physical Sciences Research Council (EPSRC) (grant no. EP/G03706X/1).
Ruth E Baker is a Royal Society Wolfson Research Merit Award holder and a Leverhulme Research Fellow, and would also like to thank the BBSRC for funding via grant number BB/R00816/1.
\end{acknowledgements}

\appendix
\section{Table of hyperparameters}
\begin{table*}[t]
  \centering
    \begin{tabular}{ | l | l | l | l | l | l | l | l | l | l | }
    \hline
    \textbf{Test} & $\mathbf{n}$& $\mathbf{A(0)}$ & $T$ & $\mathbf{\theta^*}$ & $N$ & $\mathbf{\alpha}$ & \textbf{Repeats} & \textbf{Proposal} & \textbf{Prior} \\
    \textbf{problem} & & & & & & & & \textbf{s.d.} & \textbf{interval} \\ \hline
    Toy & 32 & - & - & 10 & $5 \ast 10^4$ & 0.005 & 1 & 0.25 & $[10^{-4},10^4]$ \\
    model & & & & & & & & & \\  \hline    
    Death & 32 & 10 & 20 & $(0.1,0.01)$ & $5 \ast 10^5$ & 0.005 & 5 & 0.25 & $[10^{-4},10^4]$ \\
    process & & & & & & & & & \\  \hline
    Dimerization & 16 & $(10^5,0,0)$ & 100 & $(1,0.04,$ & $5 \ast 10^4$ & 0.05 & 1 & 0.25 & $[10^{-2},10^2],$ \\
    & & & & $ 0.002,0.5)$ & & & & & $ [10^{-3},10^1],$ \\ 
    & & & & & & & & & $ [10^{-5},10^{-1}],$ \\ 
    & & & & & & & & & $ [10^{-3},10^1]$ \\ \hline
    Diffusion & 8 & $10 \ast \mathds{1}_{x<0}$ & 20 & 0.1 & $5 \ast 10^4$ & 0.05 & 5 & 0.25 & $ [10^{-4},10^0]$ \\
    & & & & & & & & & \\  \hline
    \end{tabular}
        \caption{Summary of hyperparameters used in simulations.}
\end{table*}


\bibliography{TT18}   
\bibliographystyle{spbasic} 
%
%

\end{document}